\documentclass[a4paper,11pt]{article}

\usepackage{authblk}
\usepackage{algorithm}
\usepackage{algorithmic}
\usepackage{multirow}
\usepackage{graphicx}
\usepackage{booktabs}
\usepackage{balance}
\usepackage{epsfig}
\usepackage{epstopdf}
\usepackage{amsmath}
\usepackage{amsthm}
\usepackage{amsfonts}
\usepackage{subcaption}
\usepackage{wasysym}
\usepackage{microtype}
\usepackage{rotating}
\usepackage{tikz}
\usepackage{txfonts}
\usepackage{fancyhdr}
\usepackage{url}
\usepackage{hyperref}
\usetikzlibrary{positioning}
\usetikzlibrary{decorations.pathreplacing}
\tikzset{
	mybrace/.style={decorate,decoration={brace,aspect=#1}}
}

\newcommand{\N}{\mathbb{N}}
\newcommand{\Z}{\mathbb{Z}}
\newcommand{\F}{\mathbb{F}}

\newtheorem{definition}{Definition}
\newtheorem{lemma}{Lemma}

\newtheorem{theorem}{Theorem}

\newtheorem{remark}{Remark}
\newtheorem{problem}{Problem}
\newtheorem{example}{Example}
\newtheorem{corollary}{Corollary}

\providecommand{\keywords}[1]{\textbf{\textit{Keywords }} #1}

\begin{document}

\title{Mutually Orthogonal Latin Squares based on Cellular Automata}

\author[1]{Luca Mariot}
\author[2]{Maximilien Gadouleau}
\author[3]{Enrico Formenti}
\author[1]{Alberto Leporati}

\affil[1]{{\normalsize Dipartimento di Informatica, Sistemistica e
    Comunicazione, Universit\`{a} degli Studi di Milano-Bicocca, Viale Sarca
    336, 20126 Milano, Italy} \\

  {\small \texttt{\{luca.mariot, alberto.leporati\}@unimib.it}}}

\affil[2]{{\normalsize Department of Computer Science, Durham University, South
              Road, Durham DH1 3LE, United Kingdom} \\
  
    {\small \texttt{m.r.gadouleau@durham.ac.uk}}}
\affil[3]{{\normalsize Laboratoire d'Informatique, Signaux et Syst\`{e}mes de
              Sophia-Antipolis (I3S), Universit\'{e} C\^{o}te d'Azur,
              2000, route des Lucioles - Les Algorithmes, b\^{a}t. Euclide
              B, 06900 Sophia Antipolis (France)} \\
  
    {\small \texttt{enrico.formenti@unice.fr}}}

\maketitle

\begin{abstract}
  We investigate sets of Mutually Orthogonal Latin Squares (MOLS) generated by
  Cellular Automata (CA) over finite fields. After introducing how a CA defined
  by a bipermutive local rule of diameter $d$ over an alphabet of $q$ elements
  generates a Latin square of order $q^{d-1}$, we study the conditions under
  which two CA generate a pair of orthogonal Latin squares. In particular, we
  prove that the Latin squares induced by two Linear Bipermutive CA (LBCA) over
  the finite field $\F_q$ are orthogonal if and only if the polynomials
  associated to their local rules are relatively prime. Next, we enumerate all
  such pairs of orthogonal Latin squares by counting the pairs of coprime monic
  polynomials with nonzero constant term and degree $n$ over $\F_q$. Finally, we
  present a construction of MOLS generated by LBCA with irreducible polynomials
  and prove the maximality of the resulting sets, as well as a lower bound which
  is asymptotically close to their actual number.
\end{abstract}


\thispagestyle{fancy}

\section{Introduction}
\label{sec:intro}
A \emph{Latin square} of order $N \in \N$ is a $N\times N$ matrix where each
number from $1$ to $N$ appears exactly once in each row and column. Two Latin
squares $L_1$ and $L_2$ of order $N$ are \emph{orthogonal} if by
\emph{superimposing} them one obtains all ordered pairs $(i,j)$ of numbers from
$1$ to $N$, and \emph{Mutually Orthogonal Latin Squares} (MOLS) are sets of
Latin squares that are pairwise orthogonal.

Despite their simple definition, the construction of MOLS is a notoriously
difficult combinatorial problem and it is one of the most studied research
topics in design theory. This interest is also due to the numerous applications
that MOLS have in other fields such as cryptography (for example in the design
of authentication codes~\cite{stinson92} and
multipermutations~\cite{vaudenay94}), coding theory (see e.g. the Golomb-Posner
code~\cite{golomb64}) and statistics (particularly in the design of
experiments~\cite{montgomery17}). Some of the best known constructions of MOLS
include MacNeish's theorem~\cite{macneish22} and Wilson's
construction~\cite{wilson74} (see~\cite{keedwell15,colbourn95} for a more
complete overview of construction methods).

The goal of this paper is to investigate a new construction of MOLS based on
\emph{Cellular Automata} (CA), a particular kind of discrete dynamical system
described by a regular lattice of \emph{cells}, where each cell synchronously
updates its state by applying a local rule to itself and its neighboring
cells. The motivation for studying this construction of MOLS spawned from the
question of designing a threshold secret sharing scheme based on CA without
adjacency constraints on the shares, as in the schemes proposed
in~\cite{delrey05,mariot14}.

To this end, we first isolate a particular subclass of CA -- namely, those
defined by bipermutive local rules of diameter $d$ -- and remark that the Cayley
tables of their global rules are Latin squares of order $s^{d-1}$, where $s$ is
the size of the CA alphabet. We then narrow our attention to the case where the
local rules are linear over the finite field $\F_q$, characterizing the pairs of
Linear Bipermutive CA (LBCA) that produce orthogonal Latin squares. In particular,
we prove that the Latin squares generated by two LBCA are orthogonal if and only
if the polynomials associated to their local rules are relatively prime over
$\F_q$. This is done by observing that the orthogonality of the squares is
equivalent to the invertibility of the Sylvester matrix obtained from the
transition matrices of the LBCA. Subsequently, we determine the number of pairs
of orthogonal Latin squares generated by LBCA with rules of a fixed diameter
$d$. Due to the aforementioned characterization, this actually amounts to
counting the number of pairs of coprime monic polynomials with nonzero constant
term and degree $n=d-1$ over $\F_q$. Although the enumeration of
coprime polynomial pairs over finite fields is a well-studied
problem~\cite{reifegerste00,benjamin07}, to the best of our knowledge the
case where both polynomials have a nonzero constant term has not been
addressed before. We thus solve this counting problem through a recurrence
equation, remarking that for $q=2$ the resulting integer sequence is already
known in the OEIS for other combinatorial and number-theoretic
facts~\cite{a002450}. Finally, we present a construction of MOLS based on LBCA
whose rules are defined by the product of two irreducible polynomials, and we
prove that the size of the MOLS resulting from this construction is maximal,
meaning that they cannot be extended by adding another Latin square generated by
LBCA. Further, we count how many maximal MOLS can be produced by our
construction, and we prove that the corresponding lower bound is asymptotically
close to the actual number of maximal MOLS which can be generated by LBCA.

The present paper is an extended version of~\cite{mariot16}, a work that was
informally presented at AUTOMATA 2016. In particular, the new original
contributions of this paper concern the counting results of coprime polynomials
and the construction of MOLS based on irreducible polynomials.

The rest of this paper is organized as follows. Section~\ref{sec:basic} covers
the basic background definitions about Latin squares and cellular
automata. Section~\ref{sec:char} focuses on the characterization of orthogonal
Latin squares generated by linear bipermutive CA. Section~\ref{sec:enum}
addresses the enumeration of coprime polynomials with nonzero constant term,
which are in one-to-one correspondence with orthogonal Latin squares generated
by LBCA. Section~\ref{sec:const} describes a construction of MOLS based on LBCA
with irreducible polynomials, proves the maximality of the resulting MOLS sizes
and provides a lower bound for their number. Finally, Section~\ref{sec:conc}
summarizes the contributions of this paper, and discusses some interesting
avenues for future research on this topic.

\section{Preliminaries on Latin Squares and Cellular Automata}
\label{sec:basic}
In this section, we gather all the basic definitions that will be used to
describe our results, referring the reader to~\cite{keedwell15}
and~\cite{kari05} for further information about Latin squares and cellular
automata, respectively.

We start by giving the formal definition of a Latin square:
\begin{definition}
\label{def:ls}
Let $X$ be a finite set of cardinality $|X| = N \in \N$, and let
$[N] = \{1,\cdots, N\}$. A \emph{Latin square} of order $N$ is a $N \times N$
square matrix $L$ with entries from $X$ such that, for all $i, j, k \in [N]$
with $k \neq j$, it holds that $L(i,j) \neq L(i,k)$ and $L(j,i) \neq L(k,i)$.
\end{definition}
In other words, Definition~\ref{def:ls} states that each row and each column of
a Latin square is a permutation of the support set $X$. The concept of Latin
square is equivalent to that of \emph{quasigroup}:
\begin{definition}
\label{def:qg}
A \emph{quasigroup} of order $N \in \N$ is a pair $\langle X, \circ\rangle$
where $X$ is a finite set of $N$ elements and $\circ$ is a binary operation
over $X$ such that for all $x, y \in X$ the two equations $x \circ z = y$ and $z
\circ x = y$ admit a unique solution for all $z \in X$.
\end{definition}
Indeed, an algebraic structure $\langle X, \circ \rangle$ is a quasigroup if and
only if its \emph{Cayley table} is a Latin square~\cite{keedwell15}. In what
follows, we will assume that the support set is always $X = [N] = \{1,\cdots, N\}$.

We now introduce the orthogonality property of Latin squares:
\begin{definition}
\label{def:ols}
Two Latin squares $L_1$ and $L_2$ of order $N$ are called \emph{orthogonal Latin
  squares} (OLS) if
\begin{equation}
  (L_1(i_1,j_1),L_2(i_1,j_1)) \neq (L_1(i_2,j_2),L_2(i_2,j_2))
\end{equation}
for all distinct pairs of coordinates $(i_1,j_1), (i_2,j_2) \in [N]\times
[N]$.
\end{definition}
Equivalently, $L_1$ and $L_2$ are orthogonal if their \emph{superposition}
yields all the ordered pairs of the Cartesian product $[N] \times [N]$. A set of
$k$ Latin squares which are pairwise orthogonal is denoted as a $k$-MOLS, where the
acronym stands for \emph{Mutually Orthogonal Latin Squares}.

\emph{Cellular Automata} (CA) are a particular kind of discrete
dynamical systems defined by shift-invariant local functions. In
particular, a CA is composed of a lattice of \emph{cells} whose states
range over a finite alphabet $A$. Each cell updates in parallel its
state by applying a \emph{local rule} $f:A^{\nu} \rightarrow A$ to
itself and $\nu - 1$ surrounding cells. One of the most common studied
settings is that of one-dimensional infinite CA, where the lattice is
the full shift space $A^{\Z}$. The
\emph{Curtis-Hedlund-Theorem}~\cite{hedlund69} topologically
characterizes such CA in terms of global maps
$F: A^{\Z} \rightarrow A^{\Z}$ that are both shift-invariant and
uniformly continuous with respect to the Cantor distance.

For our work, we are interested in one-dimensional \emph{finite} CA. This case
leads to the problem of updating the cells at the boundaries, since they do not
have enough neighbors upon which the local rule can be applied. In this paper we
focus on \emph{No Boundary CA} (NBCA), which we define as follows:
\begin{definition}
\label{def:nbca}
Let $A$ be a finite alphabet and $n,d \in \N$ with $n \ge d$. The \emph{No
  Boundary Cellular Automaton} (NBCA) $F: A^{n} \rightarrow A^{n-d+1}$ of length
$n$ and diameter $d$ determined by the \emph{local rule} $f:A^d \rightarrow A$
is the vectorial function defined for all $x \in A^n$ as
\begin{equation}
    \label{eq:nbca}
      F(x_0, \cdots, x_{n-1}) = (f(x_0, \cdots, x_{d-1}), f(x_1, \cdots, x_{d}),
      \cdots, f(x_{n-d}, \cdots, x_{n-1})) \enspace .
  \end{equation}
\end{definition}
In other words, in a NBCA of length $n$, each output coordinate with
index $0\le i\le n-d$ is determined by evaluating the local rule $f$
of diameter $d$ on the \emph{neighborhood} formed by the $i$-th input
coordinate $x_i$ and the $d-1$ coordinates to its right, i.e.\
$x_{i+1},\cdots,x_{i+d}$.

The NBCA model has been investigated in~\cite{mariot19} for the design of
S-boxes. There, the authors considered the case where the alphabet is $A=\F_2$,
so that a NBCA corresponds to a particular kind of vectorial Boolean function
defined by shift-invariant coordinate functions. When the CA alphabet is $\F_2$
the local rule $f:\F_2^d \rightarrow \F_2$ can be represented by its truth
table, and its decimal representation is referred to as the \emph{Wolfram code}
of the rule. In this paper, we will mainly consider the setting where the
alphabet is the finite field $\F_q$, with $q$ being any power of a prime
number. In order to avoid burdening notation we will use CA and NBCA
interchangeably, since NBCA is the only model considered in the remainder of
this work.

The following example grounds Definition~\ref{def:nbca} for the case of binary
CA (i.e. when $A = \F_2$):
\begin{example}
  \label{ex:nbca}
  Let $A = \F_2$, and consider a CA $F:\F_2^6 \rightarrow \F_2^4$ of length
  $n=6$ and diameter $d=3$ with local rule $f: \F_2^3 \rightarrow \F_2$ defined
  as $f(x_0,x_1, x_2) = x_0 \oplus x_1 \oplus x_2$. Figure~\ref{fig:nbca}
  depicts the application of the CA global function $F$ over the vector
  $x = (0, 1, 0, 1, 0, 0)$ and reports the truth table of the local rule
  $f$. The Wolfram code of $f$ is $150$, since it corresponds to the decimal
  encoding of the output column $(0,1,1,0,1,0,0,1)$ of the table.
  \begin{figure}[tb]
	\centering
	\begin{subfigure}{0.5\textwidth}
        \centering
	\begin{tikzpicture}
	[->,auto,node distance=1.5cm, empt node/.style={font=\sffamily,inner
		sep=0pt}, rect
	node/.style={rectangle,draw,font=\bfseries,minimum size=0.5cm, inner
		sep=0pt, outer sep=0pt}]
	
	\node [empt node] (c)   {};
	\node [rect node] (c1) [right=0.1cm of c] {$1$};
	\node [rect node] (c2) [right=0cm of c1] {$0$};
	\node [rect node] (c3) [right=0cm of c2] {$1$};
	\node [rect node] (c4) [right=0cm of c3] {$1$};
	
	\node [empt node] (f1) [above=0.4cm of c2.east] {{\footnotesize
			$f(0,1,0) = 1$}};
	
	\node [rect node] (p2) [above=0.85cm of c1] {$1$};
	\node [rect node] (p1) [left=0cm of p2] {$0$};
	\node [rect node] (p3) [right=0cm of p2] {$0$};
	\node [rect node] (p4) [right=0cm of p3] {$1$};
	\node [rect node] (p5) [right=0cm of p4] {$0$};
	\node [rect node] (p6) [right=0cm of p5] {$0$};
	
	\node [empt node] (p7) [below=0.2cm of p1] {};
        \node [empt node] (p7a) [below=0.5cm of c1] {\phantom{M}};
	\node [empt node] (p8) [right=0.07cm of p7] {};
	\node [empt node] (p12) [above=0.5cm of p1.east] {};
	\node [empt node] (p13) [above=0.5cm of p5.east] {};
	\node [empt node] (p14) [above=0.3cm of p13] {\phantom{M}};
	
	\draw [-, mybrace=0.25, decorate, decoration={brace,mirror,amplitude=5pt,raise=0.3cm}]
	(p1.west) -- (p3.east) node [midway,yshift=-0.3cm] {};
	(p1.west) -- (p2.east) node [midway,yshift=0.3cm] {};
	\draw[->] (p8) -- (c1.north);
	\draw[->, draw=white] (p12) edge[bend left] (p13);
      \end{tikzpicture}
      \label{fig:nbca-a}
	\end{subfigure}%
	\begin{subfigure}{0.5\textwidth}
	  \begin{tabular}{cc}
            \hline\noalign{\smallskip}
            $x_0,x_1,x_2$ & $f(x_0,x_1,x_2)$ \\
            \noalign{\smallskip}\hline\noalign{\smallskip}
            000 & 0 \\
            100 & 1 \\
            010 & 1 \\
            110 & 0 \\
            001 & 1 \\
            101 & 0 \\
            011 & 0 \\
            111 & 1 \\
            \hline\noalign{\smallskip}
          \end{tabular}
         \label{fig:nbca-b} 
	\end{subfigure}
	\caption{Example of CA of length $n=6$ defined by rule 150.}
        \label{fig:nbca}
\end{figure}
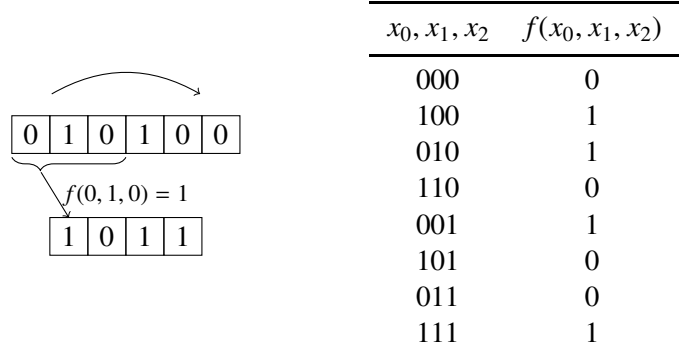
\end{example}
This paper focuses on the class of \emph{bipermutive CA}, formally defined
below:
\begin{definition}
\label{def:bca}
A CA $F:A^{n} \rightarrow A^{n-d+1}$ induced by a local rule
$f:A^{d}\rightarrow A$ is called \emph{left permutive} (respectively,
\emph{right permutive}) if, for all $z \in A^{d-1}$, the restriction
$f_{R,z}: A \rightarrow A$ (respectively, $f_{L,z}:A\rightarrow
A$) obtained by fixing the first (respectively, the last) $d-1$ coordinates
of $f$ to the values specified in $z$ is a permutation on $A$. A CA which
is both left and right permutive is said to be a \emph{bipermutive} CA (BCA).
\end{definition}
Remark that when $A = \F_2 = \{0,1\}$ a local rule $f: \F_2^d \rightarrow \F_2$
is left permutive if and only if there exists a \emph{generating function}
$\varphi: \F_2^{d-1} \rightarrow \F_2$ such that
\begin{equation}
\label{eq:perm-func}
f(x_0,x_1,\cdots,x_{d-1}) = x_0 \oplus \varphi(x_1,\cdots,x_{d-1}) \enspace ,
\end{equation}
and symmetrically for right permutive rules. Thus, bipermutive CA over $\F_2$
are those induced by local rules of the form
\begin{equation}
\label{eq:biperm-func}
f(x_0,x_1,\cdots,x_{d-1}) = x_0 \oplus \varphi(x_1,\cdots,x_{d-2}) \oplus x_{d-1} \enspace ,
\end{equation}
where $\varphi$ is a Boolean function of $d-2$ variables. Considering
Example~\ref{ex:nbca}, one can see that rule 150 is bipermutive, since it
corresponds to the case where $\varphi$ is the identity function over the
second variable of the neighborhood.

Most of the results stated in this paper concern CA that, beside being
bipermutive, are also \emph{linear} over the finite field $\F_q$. A CA
$F: \F_q^n \rightarrow \F_q^{n-d+1}$ of diameter $d$ is called linear if its
local rule $f:\F_q^d \rightarrow \F_q$ is a linear combination of the cells in
the neighborhood, i.e. there exist $a_0,\cdots,a_{d-1} \in \F_q$ such that
\begin{equation}
  \label{eq:lin-ca}
  f(x_0,\cdots,x_{d-1}) = a_0x_0 + a_1x_1 + \cdots + a_{d-1}x_{d-1} \enspace ,
\end{equation}
for all $x \in \F_q^d$, where sum and product are the field operations of
$\F_q$. For $q=2$, these respectively correspond to the logical operations XOR
($\oplus$) and AND ($\land$). A linear CA can be seen as a linear transformation
over $\F_q$-vector spaces described by the following $n \times (n-d+1)$
\emph{transition matrix}:
\begin{equation}
\label{eq:ca-matr}
M_{F} = 
\begin{pmatrix}
  a_0    & \cdots & a_{d-1} & 0 & \cdots & \cdots & \cdots & \cdots & 0 \\
  0      & a_0    & \cdots  & a_{d-1} & 0 & \cdots & \cdots & \cdots & 0 \\
  \vdots & \vdots & \vdots & \ddots  & \vdots & \vdots & \vdots & \ddots & \vdots \\
  0 & \cdots & \cdots & \cdots & \cdots & 0 & a_0 & \cdots & a_{d-1} \\
\end{pmatrix} \enspace .
\end{equation}
In particular, the CA global rule is defined as the matrix-vector multiplication
$F(x) = M_F\cdot x^{\top}$ for all $x \in \F_q^n$. As remarked
in~\cite{mariot18-naco}, the matrix $M_F$ in Equation~\eqref{eq:ca-matr} is the
generator matrix of a \emph{cyclic code}. Hence, one can naturally define the
polynomial $p_f(X) \in \F_q[X]$ associated to a linear CA $F$ as the generator
polynomial of degree $n\le d-1$ of the corresponding cyclic code:
\begin{equation}
\label{eq:pol-lin-rule}
p_f(X) = a_0 + a_1X + \cdots + a_{d-1}X^{d-1} \in \F_q[X] \enspace .
\end{equation}
It is easy to see that a linear CA is bipermutive if and only if both $a_0$ and
$a_{d-1}$ are not null. Indeed, the inverse functions of the right and left
restrictions $f_{R,z}$ and $f_{L,z}$ can be defined for all $z \in \F_q^{d-1}$
and $y \in \F_q$ as follows:
\begin{align}
  x_{d-1} &= a_{d-1}^{-1}(y - a_0z_0 - \cdots - a_{d-2}z_{d-2}) \enspace , \\
  x_0 &= a_{0}^{-1}(y - a_1z_0 - \cdots - a_{d-1}z_{d-2}) \enspace .
\end{align}
Following the notation in~\cite{mariot17-naco}, we denote by LBCA a CA $F$ which
is defined by a rule which is both linear and bipermutive. In what follows, we
will consider mainly the situation where $a_{d-1} = 1$, which means that the
polynomial $p_f(X)$ associated to the LBCA is monic of degree $n=d-1$.

\section{Characterization Results}
\label{sec:char}
In this section, we first observe that any bipermutive CA can be used to
generate a Latin square. We then prove a necessary and sufficient condition
which characterizes orthogonal Latin squares generated by pairs of LBCA.

\subsection{Latin Squares from Bipermutive CA}
\label{subsec:ls-bca}
We begin by showing that any BCA of diameter $d$ and length $2(d-1)$ generates a
Latin square of order $N = q^{d-1}$, where $q$ is the size of the CA
alphabet. To this end, we first need some additional notation and definitions.

Given an alphabet $A$ of $q$ symbols, in what follows we assume that a total
order $\le$ is defined over $A^{d-1}$ and $\phi: A^{d-1} \rightarrow [N]$ is a
monotone one-to-one mapping between $A^{d-1}$ and $[N] = \{1,\cdots,q^{d-1}\}$,
where $[N]$ is endowed with the usual order of natural numbers. We denote by
$\psi$ the inverse mapping of $\phi$.

We now formally define the notion of square associated to a CA:
\begin{definition}
\label{def:square-ca}
Let $A$ be an alphabet of $q$ symbols. The \emph{square} associated to the CA
$F: A^{2(d-1)} \rightarrow A^{d-1}$ defined by rule $f: A^d \rightarrow A$ is
the square matrix $\mathcal{S}_F$ of size $q^{d-1} \times q^{d-1}$ with entries
from $[q^{d-1}]$ defined for all $1 \le i,j \le q^{d-1}$ as
\begin{equation}
\label{eq:sq-ca}
\mathcal{S}_{F}(i,j) = \phi(F(\psi(i)||\psi(j))) \enspace ,
\end{equation}
where $\psi(i)||\psi(j) \in A^{2(d-1)}$ denotes the \emph{concatenation} of
$\psi(i),\psi(j) \in A^{d-1}$.
\end{definition}
Hence, the square $\mathcal{S}_{F}$ is defined by encoding the first half of the
CA configuration as the row coordinate $i$, the second half as the column
coordinate $j$ and the output $F(\psi(i)||\psi(j))$ as the entry at $(i,j)$.

As an example, for $A = \F_2$ and diameter $d=3$, Figure~\ref{fig:r150-sq}
depicts the square $\mathcal{S}_F$ associated to the CA
$F: \F_2^4 \rightarrow \F_2^2$ defined by rule $150$. The mapping $\phi$ is
defined as $\phi(00) \mapsto 1$, $\phi(10) \mapsto 2$, $\phi(01) \mapsto 3$ and
$\phi(11) \mapsto 4$. Notice that in this particular case $\mathcal{S}_F$ is a
Latin square.

\begin{figure}[t]
\centering
\begin{subfigure}{.5\textwidth}
\centering
\begin{tikzpicture}
[->,auto,node distance=1.5cm,
       empt node/.style={font=\sffamily,inner sep=0pt,minimum size=0pt},
       rect node/.style={rectangle,draw,font=\sffamily,minimum size=0.3cm, inner sep=0pt, outer sep=0pt}]
	
        \node [empt node] (e1) {};
	\node [rect node] (i111) [right=0.5cm of e1] {$0$};
        \node [rect node] (i112) [right=0cm of i111] {$0$};
        \node [rect node] (i113) [right=0cm of i112] {$0$};
        \node [rect node] (i114) [right=0cm of i113] {$0$};
        \node [rect node] (i115) [below=0cm of i112] {$0$};
        \node [rect node] (i116) [right=0cm of i115] {$0$};

        \node [rect node] (i121) [right=0.3cm of i114] {$0$};
        \node [rect node] (i122) [right=0cm of i121] {$0$};
        \node [rect node] (i123) [right=0cm of i122] {$1$};
        \node [rect node] (i124) [right=0cm of i123] {$0$};
        \node [rect node] (i125) [below=0cm of i122] {$1$};
        \node [rect node] (i126) [right=0cm of i125] {$1$};

        \node [rect node] (i131) [right=0.3cm of i124] {$0$};
        \node [rect node] (i132) [right=0cm of i131] {$0$};
        \node [rect node] (i133) [right=0cm of i132] {$0$};
        \node [rect node] (i134) [right=0cm of i133] {$1$};
        \node [rect node] (i135) [below=0cm of i132] {$0$};
        \node [rect node] (i136) [right=0cm of i135] {$1$};

        \node [rect node] (i141) [right=0.3cm of i134] {$0$};
        \node [rect node] (i142) [right=0cm of i141] {$0$};
        \node [rect node] (i143) [right=0cm of i142] {$1$};
        \node [rect node] (i144) [right=0cm of i143] {$1$};
        \node [rect node] (i145) [below=0cm of i142] {$1$};
        \node [rect node] (i146) [right=0cm of i145] {$0$};

	\node [rect node] (i211) [below=0.6cm of i111] {$1$};
        \node [rect node] (i212) [right=0cm of i211] {$0$};
        \node [rect node] (i213) [right=0cm of i212] {$0$};
        \node [rect node] (i214) [right=0cm of i213] {$0$};
        \node [rect node] (i215) [below=0cm of i212] {$1$};
        \node [rect node] (i216) [right=0cm of i215] {$0$};

        \node [rect node] (i221) [right=0.3cm of i214] {$1$};
        \node [rect node] (i222) [right=0cm of i221] {$0$};
        \node [rect node] (i223) [right=0cm of i222] {$1$};
        \node [rect node] (i224) [right=0cm of i223] {$0$};
        \node [rect node] (i225) [below=0cm of i222] {$0$};
        \node [rect node] (i226) [right=0cm of i225] {$1$};

        \node [rect node] (i231) [right=0.3cm of i224] {$1$};
        \node [rect node] (i232) [right=0cm of i231] {$0$};
        \node [rect node] (i233) [right=0cm of i232] {$0$};
        \node [rect node] (i234) [right=0cm of i233] {$1$};
        \node [rect node] (i235) [below=0cm of i232] {$1$};
        \node [rect node] (i236) [right=0cm of i235] {$1$};

        \node [rect node] (i241) [right=0.3cm of i234] {$1$};
        \node [rect node] (i242) [right=0cm of i241] {$0$};
        \node [rect node] (i243) [right=0cm of i242] {$1$};
        \node [rect node] (i244) [right=0cm of i243] {$1$};
        \node [rect node] (i245) [below=0cm of i242] {$0$};
        \node [rect node] (i246) [right=0cm of i245] {$0$};

	\node [rect node] (i311) [below=0.6cm of i211] {$0$};
        \node [rect node] (i312) [right=0cm of i311] {$1$};
        \node [rect node] (i313) [right=0cm of i312] {$0$};
        \node [rect node] (i314) [right=0cm of i313] {$0$};
        \node [rect node] (i315) [below=0cm of i312] {$1$};
        \node [rect node] (i316) [right=0cm of i315] {$1$};

        \node [rect node] (i321) [right=0.3cm of i314] {$0$};
        \node [rect node] (i322) [right=0cm of i321] {$1$};
        \node [rect node] (i323) [right=0cm of i322] {$1$};
        \node [rect node] (i324) [right=0cm of i323] {$0$};
        \node [rect node] (i325) [below=0cm of i322] {$0$};
        \node [rect node] (i326) [right=0cm of i325] {$0$};

        \node [rect node] (i331) [right=0.3cm of i324] {$0$};
        \node [rect node] (i332) [right=0cm of i331] {$1$};
        \node [rect node] (i333) [right=0cm of i332] {$0$};
        \node [rect node] (i334) [right=0cm of i333] {$1$};
        \node [rect node] (i335) [below=0cm of i332] {$1$};
        \node [rect node] (i336) [right=0cm of i335] {$0$};

        \node [rect node] (i341) [right=0.3cm of i334] {$0$};
        \node [rect node] (i342) [right=0cm of i341] {$1$};
        \node [rect node] (i343) [right=0cm of i342] {$1$};
        \node [rect node] (i344) [right=0cm of i343] {$1$};
        \node [rect node] (i345) [below=0cm of i342] {$0$};
        \node [rect node] (i346) [right=0cm of i345] {$1$};

	\node [rect node] (i411) [below=0.6cm of i311] {$1$};
        \node [rect node] (i412) [right=0cm of i411] {$1$};
        \node [rect node] (i413) [right=0cm of i412] {$0$};
        \node [rect node] (i414) [right=0cm of i413] {$0$};
        \node [rect node] (i415) [below=0cm of i412] {$0$};
        \node [rect node] (i416) [right=0cm of i415] {$1$};

        \node [rect node] (i421) [right=0.3cm of i414] {$1$};
        \node [rect node] (i422) [right=0cm of i421] {$1$};
        \node [rect node] (i423) [right=0cm of i422] {$1$};
        \node [rect node] (i424) [right=0cm of i423] {$0$};
        \node [rect node] (i425) [below=0cm of i422] {$1$};
        \node [rect node] (i426) [right=0cm of i425] {$0$};

        \node [rect node] (i431) [right=0.3cm of i424] {$1$};
        \node [rect node] (i432) [right=0cm of i431] {$1$};
        \node [rect node] (i433) [right=0cm of i432] {$0$};
        \node [rect node] (i434) [right=0cm of i433] {$1$};
        \node [rect node] (i435) [below=0cm of i432] {$0$};
        \node [rect node] (i436) [right=0cm of i435] {$0$};

        \node [rect node] (i441) [right=0.3cm of i434] {$1$};
        \node [rect node] (i442) [right=0cm of i441] {$1$};
        \node [rect node] (i443) [right=0cm of i442] {$1$};
        \node [rect node] (i444) [right=0cm of i443] {$1$};
        \node [rect node] (i445) [below=0cm of i442] {$1$};
        \node [rect node] (i446) [right=0cm of i445] {$1$};
	
\end{tikzpicture}
\end{subfigure}%
\begin{subfigure}{.5\textwidth}
\centering
\begin{tikzpicture}
[->,auto,node distance=1.5cm,
       empt node/.style={font=\sffamily,inner sep=0pt,minimum size=0pt},
       rect node/.style={rectangle,draw,font=\sffamily,minimum size=0.8cm, inner sep=0pt, outer sep=0pt}]
	\large
        
	\node [rect node] (s11) {$1$};
        \node [rect node] (s12) [right=0cm of s11] {$4$};
        \node [rect node] (s13) [right=0cm of s12] {$3$};
        \node [rect node] (s14) [right=0cm of s13] {$2$};

	\node [rect node] (s21) [below=0cm of s11] {$2$};
        \node [rect node] (s22) [right=0cm of s21] {$3$};
        \node [rect node] (s23) [right=0cm of s22] {$4$};
        \node [rect node] (s24) [right=0cm of s23] {$1$};

	\node [rect node] (s31) [below=0cm of s21] {$4$};
        \node [rect node] (s32) [right=0cm of s31] {$1$};
        \node [rect node] (s33) [right=0cm of s32] {$2$};
        \node [rect node] (s34) [right=0cm of s33] {$3$};

	\node [rect node] (s41) [below=0cm of s31] {$3$};
        \node [rect node] (s42) [right=0cm of s41] {$2$};
        \node [rect node] (s43) [right=0cm of s42] {$1$};
        \node [rect node] (s44) [right=0cm of s43] {$4$};
	
\end{tikzpicture}
\end{subfigure}%
\caption{Example of square of order $2^{3-1}=4$ induced by rule $150$.}
\label{fig:r150-sq}
\end{figure}
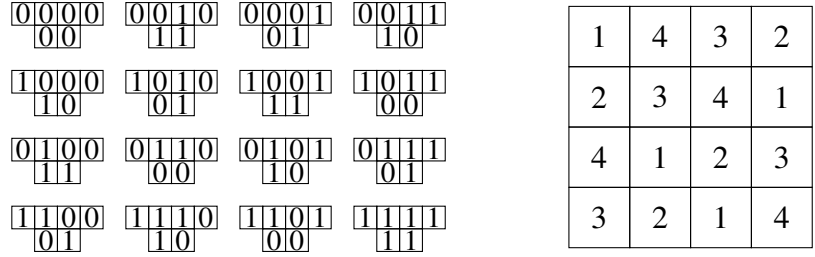

We remark that this representation has been adopted in several works in the CA
literature, even though under a different guise. Indeed, one can consider the
square associated to a CA as the Cayley table of an algebraic structure
$\langle S, \circ \rangle$, where $S$ is a set of size $|A|^{d-1}$ isomorphic to
$A^{d-1}$, and $\circ$ is a binary operation over $S$. The two operands
$x, y \in S$ are represented by the vectors respectively composed of the
leftmost and rightmost $d-1$ input cells of the CA, while the $d-1$ output cells
represent the result $z = x \circ y$. To the best of our knowledge, the first
researchers who employed this algebraic characterization of cellular automata
were Pedersen~\cite{pedersen92} and Eloranta~\cite{eloranta93}, respectively for
investigating their periodicity and partial reversibility properties. Other
works in this line of research include Moore and Drisko~\cite{moore96}, who
studied the algebraic properties of the square representation of CA, and
Moore~\cite{moore98}, who considered the computational complexity of predicting
CA whose local rules define solvable and nilpotent groups.

As noticed above, the square associated to the CA defined by rule 150 is
actually a Latin square. We will now show that this holds in general for all
bipermutive CA. To this end, we first recall a Lemma proved in~\cite{mariot14},
which states that fixing $d-1$ adjacent cells in a bipermutive CA yields a
permutation between the remaining variables and the output:
\begin{lemma}
\label{lm:restr-perm}
Let $F: A^n \rightarrow A^{n-d+1}$ be a BCA defined by local rule
$f: A^{d} \rightarrow A$. Given $\tilde{x} \in A^{d-1}$ and $i$ with
$0 \le i \le n-d+1$, let $F|_{\tilde{x},i}: A^{n-d+1} \rightarrow A^{n-d+1}$ be
the restriction of $F$ obtained by fixing to $\tilde{x}$ the block of $d-1$
consecutive coordinates starting in $i$ of the BCA input vector, i.e.
$x_i = \tilde{x}_0$, $x_{i+1} = \tilde{x}_1$, $\cdots$,
$x_{i+d-2} = \tilde{x}_{d-2}$. Then, $F|_{\tilde{x},i}$ is a permutation over
$A^{n-d+1}$.
\end{lemma}
On account of Lemma~\ref{lm:restr-perm}, we can prove that the squares
associated to bipermutive CA are indeed Latin squares:
\begin{lemma}
\label{lm:lat-sq-bip-ca}
Let $A$ be an alphabet of $q$ symbols, and $d\ge 2$. Then, the square $L_{F}$ of
the BCA $F: A^{2(d-1)} \rightarrow A^{d-1}$ defined by local rule
$f:A^{d}\rightarrow A$ is a Latin square of order $N=q^{d-1}$.
\end{lemma}
\begin{proof}
  Let $i \in [N]$ be a row of $L_{F}$, and let
  $\psi(i) = (x_0,\cdots,x_{d-2}) \in A^{d-1}$ be the vector associated to $i$
  with respect to the total order $\le$ on $A^{d-1}$. Consider now the set
  $C = \{c \in A^{2(d-1)}: (c_0,\cdots,c_{d-2}) = \psi(i)\}$, i.e. the set of
  configurations of length $2(d-1)$ whose first $d-1$ coordinates coincide with
  $\psi(i)$, and let $F_{\psi(i),0}: A^{d-1} \rightarrow A^{d-1}$ be the
  restriction of $F$ determined by $\psi(i)$. By Lemma~\ref{lm:restr-perm}, the
  function $F_{\psi(i),0}$ is a permutation over $A^{d-1}$. So, the $i$-th row
  of $L_{F}$ is a permutation of $[N]$. A symmetric argument holds for any
  column $j$ of $L_{F}$, with $1\le j \le N$, which fixes the rightmost $d-1$
  variables of $F$ to $\psi(j)$. Hence, every column of $L_{F}$ is also a
  permutation of $[N]$, and thus $L_{F}$ is a Latin square of order
  $q^{d-1}$.\qed
\end{proof}

\subsection{Orthogonal Latin Squares from Linear Bipermutive CA}
\label{subsec:ols-lin-bca}
In the next result, we prove a characterization of orthogonal Latin squares
generated by LBCA in terms of their associated polynomials:

\begin{theorem}
  \label{thm:ols-lbca}
  Let $F, G:\F_q^{2(d-1)} \rightarrow \F_q^{d-1}$ be two LBCA of length
  $2(d-1)$, respectively defined by the local rules
  $f,g: \F_q^d \rightarrow \F_q$ defined as:
  \begin{align}
    f(x_0,\cdots,x_{d-1}) &= a_0x_0 + \cdots + a_{d-1}x_{d-1} \enspace , \\
    g(x_0,\cdots,x_{d-1}) &= b_0x_0 + \cdots + b_{d-1}x_{d-1} \enspace .
  \end{align}
  Then, the Latin squares $L_{F}$ and $L_{G}$ of order $q^{d-1}$ generated by
  $F$ and $G$ are orthogonal if and only if the polynomials
  $p_f(X),p_g(X) \in \F_q[X]$ associated to $f$ and $g$ are relatively prime.
\end{theorem}
\begin{proof}
  Denote by $z=x||y$ the concatenation of vectors $x$ and $y$. We show that the
  function
  $\mathcal{H}: \F_q^{2(d-1)}\times \F_q^{2(d-1)} \rightarrow
  \F_q^{2(d-1)}\times\F_q^{2(d-1)}$, defined for all
  $(x,y) \in \F_q^{2(d-1)}\times \F_q^{2(d-1)}$ as
  \begin{equation}
  \label{eq:map-h}
  \mathcal{H}(x,y) = (F(z), \mathcal{G}(z)) = (\tilde{x}, \tilde{y})
  \end{equation}
  is bijective if and only if the polynomials $p_f(X)$ and $p_g(X)$ associated
  to $F$ and $G$ are coprime. Given the transition matrices $M_{F}$ and $M_{G}$
  respectively associated to $F$ and $G$, one can rewrite
  Equation~\eqref{eq:map-h} as a system of two equations:
  \begin{equation}
    \label{eq:system-1}
    \begin{cases}
      F(z) = M_{F}z^\top = \tilde{x} \\
      \mathcal{G}(z) = M_{G}z^\top = \tilde{y}
    \end{cases} \enspace .
  \end{equation}
  Since both $M_F$ and $M_G$ have size $(d-1)\times 2(d-1)$,
  Equation~\eqref{eq:system-1} is a linear system of $2(d-1)$ equations and
  $2(d-1)$ unknowns, defined by the following $2(d-1)\times 2(d-1)$ square
  matrix:
  \begin{equation}
    \label{eq:sylv-matr}
    M =
    \begin{pmatrix}
      a_0    & \cdots & a_{d-1} & 0 & \cdots & \cdots & \cdots & \cdots & 0 \\
      0      & a_0    & \cdots  & a_{d-1} & 0 & \cdots & \cdots & \cdots & 0 \\
      \vdots & \vdots & \vdots & \ddots  & \vdots & \vdots & \vdots & \ddots & \vdots \\
      0 & \cdots & \cdots & \cdots & \cdots & 0 & a_0 & \cdots & a_{d-1} \\
      b_0    & \cdots & b_{d-1} & 0 & \cdots & \cdots & \cdots & \cdots & 0 \\
      0      & b_0    & \cdots  & b_{d-1} & 0 & \cdots & \cdots & \cdots & 0 \\
      \vdots & \vdots & \vdots & \ddots  & \vdots & \vdots & \vdots & \ddots & \vdots \\
      0 & \cdots & \cdots & \cdots & \cdots & 0 & b_0 & \cdots & b_{d-1} \\
\end{pmatrix} \enspace ,
\end{equation} 
i.e., $M$ is obtained by superposing the transition matrices $M_F$ and
$M_G$. Thus $\mathcal{H}(x,y) = Mz^\top$ and $\mathcal{H}$ is bijective if and
only if the determinant of $M$ is not null. Remark that matrix $M$ in
Equation~\eqref{eq:sylv-matr} is a \emph{Sylvester matrix}, and its determinant
is the \emph{resultant} of the two polynomials $p_f(X)$ and $p_g(X)$ associated
to the LBCA $F$ and $G$, respectively. It is well known (see for
instance~\cite{lidl-field}) that the resultant of two polynomials is nonzero if
and only if they are relatively prime. Hence, $\mathcal{H}$ is bijective (or
equivalently, the Latin squares $L_F$ and $L_G$ are orthogonal) if and only if
the polynomials $p_f(X)$ and $p_g(X)$ are relatively prime.\qed
\end{proof}
The next result immediately follows from the above theorem:
\begin{corollary}
\label{cor:mols-lbca}
A family $p_1(X),\cdots,p_k(X) \in \F_q[X]$ of $k \in \N$ pairwise coprime
polynomials of degree $n=d-1$ is equivalent to a set of $k$ MOLS of order
$q^{n}$ generated by LBCA.
\end{corollary}

For alphabet $A=\F_2$ and diameter $d=3$ there exist only two linear bipermutive
rules, i.e. rule $150$ and rule $90$, the latter defined as
$f_{90}(x_0,x_1,x_2) = x_0 \oplus x_2$. As shown in Figure~\ref{fig:r150-90},
the Latin squares of order $N=4$ defined by the LBCA $F_{150}$ and $F_{90}$
respectively induced by $f_{150}$ and $f_{90}$ are orthogonal, since the
associated polynomials $p_{150}(X) = 1+X+X^2$ and $p_{90}(X) = 1+X^2$ are
coprime over $\F_2$.
\begin{figure}[t]
\centering
\begin{subfigure}{.3\textwidth}
\centering
\begin{tikzpicture}
[->,auto,node distance=1.5cm,
       empt node/.style={font=\sffamily,inner sep=0pt,minimum size=0pt},
       rect node/.style={rectangle,draw,font=\sffamily,minimum size=0.7cm, inner sep=0pt, outer sep=0pt}]

	\node [rect node] (s11) {$1$};
        \node [rect node] (s12) [right=0cm of s11] {$4$};
        \node [rect node] (s13) [right=0cm of s12] {$3$};
        \node [rect node] (s14) [right=0cm of s13] {$2$};

	\node [rect node] (s21) [below=0cm of s11] {$2$};
        \node [rect node] (s22) [right=0cm of s21] {$3$};
        \node [rect node] (s23) [right=0cm of s22] {$4$};
        \node [rect node] (s24) [right=0cm of s23] {$1$};

	\node [rect node] (s31) [below=0cm of s21] {$4$};
        \node [rect node] (s32) [right=0cm of s31] {$1$};
        \node [rect node] (s33) [right=0cm of s32] {$2$};
        \node [rect node] (s34) [right=0cm of s33] {$3$};

	\node [rect node] (s41) [below=0cm of s31] {$3$};
        \node [rect node] (s42) [right=0cm of s41] {$2$};
        \node [rect node] (s43) [right=0cm of s42] {$1$};
        \node [rect node] (s44) [right=0cm of s43] {$4$};
	
\end{tikzpicture}
\caption{Rule $150$}
\end{subfigure}%
\begin{subfigure}{.3\textwidth}
\centering
\begin{tikzpicture}
[->,auto,node distance=1.5cm,
       empt node/.style={font=\sffamily,inner sep=0pt,minimum size=0pt},
       rect node/.style={rectangle,draw,font=\sffamily,minimum size=0.7cm, inner sep=0pt, outer sep=0pt}]

	\node [rect node] (s11) {$1$};
        \node [rect node] (s12) [right=0cm of s11] {$2$};
        \node [rect node] (s13) [right=0cm of s12] {$3$};
        \node [rect node] (s14) [right=0cm of s13] {$4$};

	\node [rect node] (s21) [below=0cm of s11] {$2$};
        \node [rect node] (s22) [right=0cm of s21] {$1$};
        \node [rect node] (s23) [right=0cm of s22] {$4$};
        \node [rect node] (s24) [right=0cm of s23] {$3$};

	\node [rect node] (s31) [below=0cm of s21] {$3$};
        \node [rect node] (s32) [right=0cm of s31] {$4$};
        \node [rect node] (s33) [right=0cm of s32] {$1$};
        \node [rect node] (s34) [right=0cm of s33] {$2$};

	\node [rect node] (s41) [below=0cm of s31] {$4$};
        \node [rect node] (s42) [right=0cm of s41] {$3$};
        \node [rect node] (s43) [right=0cm of s42] {$2$};
        \node [rect node] (s44) [right=0cm of s43] {$1$};
	
\end{tikzpicture}
\caption{Rule $90$}
\end{subfigure}%
\begin{subfigure}{.3\textwidth}
\centering
\begin{tikzpicture}
[->,auto,node distance=1.5cm,
       empt node/.style={font=\sffamily,inner sep=0pt,minimum size=0pt},
       rect node/.style={rectangle,draw,font=\sffamily,minimum size=0.7cm, inner sep=0pt, outer sep=0pt}]

	\node [rect node] (s11) {1,1};
        \node [rect node] (s12) [right=0cm of s11] {$4,2$};
        \node [rect node] (s13) [right=0cm of s12] {$3,3$};
        \node [rect node] (s14) [right=0cm of s13] {$2,4$};

	\node [rect node] (s21) [below=0cm of s11] {$2,2$};
        \node [rect node] (s22) [right=0cm of s21] {$3,1$};
        \node [rect node] (s23) [right=0cm of s22] {$4,4$};
        \node [rect node] (s24) [right=0cm of s23] {$1,3$};

	\node [rect node] (s31) [below=0cm of s21] {$4,3$};
        \node [rect node] (s32) [right=0cm of s31] {$1,4$};
        \node [rect node] (s33) [right=0cm of s32] {$2,1$};
        \node [rect node] (s34) [right=0cm of s33] {$3,2$};

	\node [rect node] (s41) [below=0cm of s31] {$3,4$};
        \node [rect node] (s42) [right=0cm of s41] {$2,3$};
        \node [rect node] (s43) [right=0cm of s42] {$1,2$};
        \node [rect node] (s44) [right=0cm of s43] {$4,1$};
	
\end{tikzpicture}
\caption{Overlay}
\end{subfigure}%
\caption{Orthogonal Latin squares generated by BCA with rules 150 and 90,
  corresponding to the pair of coprime polynomials $1+X+X^2$ and $1+X^2$.}
\label{fig:r150-90}
\end{figure}

\section{Counting Coprime Polynomial Pairs}
\label{sec:enum}
By Corollary~\ref{cor:mols-lbca}, one can generate a set of $k$ MOLS of order
$q^{d-1}$ through LBCA of diameter $d$ by finding $k$ pairwise relatively prime
polynomials of degree $n=d-1$. The problem of counting the number of pairs of
relatively prime polynomials over finite fields has been investigated in several
papers (see e.g.~\cite{reifegerste00,allender03,benjamin07,hou09}). However,
notice that determining the number of pairs of linear CA inducing orthogonal
Latin squares entails counting only specific pairs of polynomials, namely those
whose constant term is not null. This is due to the requirement that the CA
local rules must be bipermutive. To the best of our knowledge, this particular
counting problem has not been considered in the literature, for which reason we
address it in this section.

Formally, for $n \ge 1$ let
\begin{equation}
  \label{eq:set-pol-1}
	S_n = \{ f(X) = a_0 + a_1 X + \dots + a_{n-1}X^{n-1} + X^n : a_0 \ne 0 \}
\end{equation}
be the set of monic polynomials in $\F_q[X]$ of degree $n$ and with nonzero
constant term. For all $n \ge 1$ we have that $s_n = |S_n| = (q-1)
q^{n-1}$. Moreover, we define $S_0 = \{1\}$ (the unique monic polynomial of
degree zero), and hence $s_0 = 1$.

Recall that the greatest common divisor of two polynomials $f, g \in \F_q[X]$ is
the unique monic polynomial of highest degree $h$ such that
\begin{align*}
	f(X) &= h(X) i(X) \enspace , \\
	g(X) &= h(X) j(X) \enspace ,
\end{align*}
for some $i, j \in \F_q[X]$. We remark that if $f, g \in S_n$, then
$i,j \in S_{e}$ for some $0 \le e \le n$ and $h \in S_{n-e}$.

Additionally, let us define the following subsets of $S_n^2 = S_n \times S_n$:
\begin{align*}
	A_n = \{ (f,g) \in S_n^2 : \gcd(f,g) = 1 \} \enspace , \\
	B_n = \{ (f,g) \in S_n^2 : \gcd(f,g) \ne 1 \} \enspace .
\end{align*}
In other words, $A_n$ and $B_n$ are respectively the sets of pairs of coprime
and non-coprime monic polynomials of degree $n$ with nonzero constant
term. Similarly as above, let $a_n = |A_n|$ and $b_n = |B_n|$. We are interested
in determining $a_n$, since by Theorem~\ref{thm:ols-lbca} the cardinality of
$A_n$ corresponds to the number of orthogonal Latin squares of order $q^{n}$
generated by LBCA pairs of diameter $d=n+1$. For $n=0$, we clearly have
$a_0 = 1$ and $b_0 = 0$. The following result characterizes $a_n$ for all
$n\ge 1$:
\begin{theorem}
  \label{thm:num-coprime}
  Let $n \ge 1$. Then, the number of pairs of coprime monic polynomials of
  degree $n$ with nonzero constant term is
  \begin{equation}
    \label{eq:num-coprime}
    a_n = q (q-1)^3 \frac{q^{2n-2} - 1}{q^2 - 1} + (q-1)(q-2) \enspace .
  \end{equation}
\end{theorem}
\begin{proof}
  Let us first settle the case $n=1$, for which we have
\begin{displaymath}
  S_1 = \{ f(X) = a_0 + X : a_0 \ne 0 \} \enspace .
\end{displaymath}
It is clear that $\gcd(f,g) = 1$ for any $f, g \in S_1$ if and only if
$f \ne g$. Thus, it follows that
\begin{displaymath}
  a_1 = (q-1)(q-2) \enspace ,
\end{displaymath}
which proves Equation~\eqref{eq:num-coprime} for the case $n=1$. For $n>1$,
remark first that
\begin{equation}
  \label{eq:sum-cop-ncop}
  s_n^2 = a_n + b_n \enspace .
\end{equation}
Moreover, any pair $(f,g)$ with $\deg( \gcd(f,g) ) = n-e$ ($0 \le e \le n-1$)
can be uniquely expressed as a pair $(h, (i,j))$, where $h \in S_{n-e}$ and
$(i,j) \in S_e$. Hence,
\begin{equation}
  \label{eq:ncop}
  b_n = \sum_{e=0}^{n-1} a_e s_{n-e} \enspace .
\end{equation}
Combining Equations~\eqref{eq:sum-cop-ncop} and~\eqref{eq:ncop} we have
\begin{align}
  \label{eq:an-0}
  a_n &= s_n^2 - \sum_{e=0}^{n-1} a_e s_{n-e} \enspace , \\
  \label{eq:an-1}
  a_{n-1} &= s_{n-1}^2 - \sum_{e=0}^{n-2} a_e s_{n-1-e} \enspace .
\end{align}
Multiplying both sides of~\eqref{eq:an-1} by $q$ we obtain
\begin{equation}
  \label{eq:qan-1}
  q a_{n-1} = q s_{n-1}^2 - \sum_{e=0}^{n-2} a_e (q s_{n-1-e}) \enspace .
\end{equation}
Remark that $qs_{n-1-e} = q(q-1)q^{n-2-e} = (q-1)q^{n-1-e} = s_{n-e}$. Hence,
Equation~\eqref{eq:qan-1} can be rewritten as
\begin{equation}
  \label{eq:qan-1-2}
  q a_{n-1} = qs_{n-1}^2 - \sum_{e=0}^{n-2} a_e s_{n-e} \enspace .
\end{equation}
By subtracting Equations~\eqref{eq:an-0} and~\eqref{eq:qan-1-2} we thus have
\begin{align}
  \nonumber
  a_n - q a_{n-1} &= s_n^2 - \sum_{e=0}^{n-1} a_e s_{n-e} - q s_{n-1}^2 +
                    \sum_{e=0}^{n-2} a_e s_{n-e} \\
  \label{eq:an-qan-1}
                  &= s_n^2 - q s_{n-1}^2 - a_{n-1}s_1 \enspace .
\end{align}
Since $s_n^2 = (q-1)^2q^{2n-2}$, while $qs_{n-1}^2 = (q-1)^2q^{2n-3}$ and
$s_1 = (q-1)$, Equation~\eqref{eq:an-qan-1} becomes
\begin{equation}
  a_n - q a_{n-1} = (q-1)^2(q^{2n-2} - q^{2n-3}) - a_{n-1}(q-1) \enspace ,
\end{equation}
from which it follows that
\begin{align}
  \nonumber
  a_n &= (q-1)^2 (q^{2n-2} - q^{2n-3}) + a_{n-1} \\
  \nonumber
      &= (q-1)^2 q^{2n-3} (q-1) + a_{n-1} \\
  \label{eq:an-qan-1-3}
      &= (q-1)^3 q q^{2n-4} + a_{n-1} \enspace .
\end{align}
By iterating the above procedure to determine the term $a_{n-1}$ from the
difference $a_{n-1}-qa_{n-2}$, the term $a_{n-2}$ from $a_{n-2} - qa_{n-3}$, and
so on until $a_2$ from $a_2 - qa_1$, one has
\begin{align}
  \nonumber
  a_n &= q (q-1)^3 \{q^{2n-4} + q^{2n-6} + \cdots + q^2 + 1\} + a_1 \\
  \nonumber
      &= q (q-1)^3 \sum_{t=0}^{n-2}q^{2t} + a_1 \\
  \nonumber
      &= q (q-1)^3 \frac{q^{2n-2} - 1}{q^2 - 1} + a_1 \\
  \label{eq:an-fin}
      &= q (q-1)^3 \frac{q^{2n-2} - 1}{q^2 - 1} + (q-1)(q-2) \enspace ,
\end{align}
from which we finally obtain the result.\qed
\end{proof}

\begin{remark}
  \label{rem:oeis}
  Notice that in Theorem~\ref{thm:num-coprime} we count all \emph{ordered}
  coprime polynomial pairs. To get the number of \emph{distinct} pairs, one
  simply has to divide Equation~\eqref{eq:num-coprime} by $2$, thus obtaining
  $\tilde{a}_n = \frac{1}{2}a_n$. In particular, for $q=2$ the formula for
  $\tilde{a}_n$ becomes
  \begin{equation}
    \label{eq:numb-cop-true}
    \tilde{a}_n = \frac{4^{n-1}-1}{3} \enspace .
  \end{equation}
  The first terms of this sequence for $n\ge 1$ are:
  \begin{equation}
    \label{eq:seq-oeis}
    \tilde{a}_n = 0, 1, 5, 21, 85, 341, 1365, \cdots
  \end{equation}
  which is a shifted version of OEIS sequence {\sc A002450}~\cite{a002450},
  defined by
\begin{equation}
\label{eq:a002450}
c_n = \frac{4^n - 1}{3} \enspace .
\end{equation}
It is easily seen that $c_n = \tilde{a}_{n+1}$, i.e. $c_n$ corresponds to the
number of distinct coprime pairs of polynomials of degree $n+1$ over $\F_2$
where both polynomials have nonzero constant term. We remark that sequence {\sc
  A002450} is known for several other combinatorial facts not related to
polynomials or orthogonal Latin squares arising from LBCA, for which we refer
the reader to~\cite{a002450}.
\end{remark}

\section{A Construction of MOLS based on LBCA}
\label{sec:const}
In this section, we tackle the question of determining the maximum number of
MOLS generated by linear bipermutive CA over $\F_q$ of a given order. Given
$n \in \N$ we consider in particular the following two problems:
\begin{problem}
\label{prob:max-mols}
What is the maximum number $N_{n}$ of LBCA over $\F_q$ of diameter $n+1$ whose
Latin squares are mutually orthogonal? From Section~\ref{sec:char}, this
actually amounts to compute the maximum number of monic pairwise coprime
polynomials of degree $n$ and nonzero constant term over $\F_q$.
\end{problem}
\begin{problem}
  \label{prob:numb-max-mols}
  What is the number $T_{n}$ of maximal sets of $N_{n}$ MOLS generated by
  LBCA?
\end{problem}
In the remainder of this section we present a construction for sets of MOLS
based on LBCA defined by pairwise coprime polynomials over $\F_q$. Moreover, we
solve Problem~\ref{prob:max-mols} by proving that the size of MOLS resulting
from this construction is maximal. We also determine the number $D_n$ of MOLS
that can be generated through this construction, and show that it is
asymptotically close to $T_n$.

Recall from Section~\ref{sec:enum} that $S_n$ denotes the set of all degree $n$
monic polynomials $f \in \F_q[X]$ with nonzero constant term
$a_0$. Additionally, let
\begin{equation}
\mathcal{M}_{n} = \{ R_{n} \subseteq S_{n}: \forall f\neq g \in
R_{n}, \gcd(f,g) = 1\} \enspace .
\end{equation}
In other words, $\mathcal{M}_{n}$ is the family of subsets of $S_{n}$ of
pairwise coprime polynomials. In order to solve Problem~\ref{prob:max-mols}, we
have to determine the maximal cardinality of the subsets in $\mathcal{M}_{n}$,
that is
\begin{equation}
N_{n} = \max_{R_{n} \in \mathcal{M}_{n}} |R_{n}| \enspace .
\end{equation}
On the other hand, for Problem~\ref{prob:numb-max-mols} we want to count how
many sets in $\mathcal{M}_{n}$ have cardinality $N_{n}$:
\begin{equation}
T_{n} = |\{R_{n} \in \mathcal{M}_{n} : |R_{n}| = N_{n}\}| \enspace .
\end{equation}
We begin by considering the set $\mathcal{I}_{n}$ of irreducible polynomials of
degree $n$ over $\F_q$ with nonzero constant term, all of which are trivially
pairwise coprime. Hence, $\mathcal{I}_{n}$ is included in all subsets having
maximum cardinality $N_{n}$. Denoting by $I_n$ the cardinality of
$\mathcal{I}_n$, one has that $I_0 = 1$ and $I_1 = q-1$, while for $n\ge 2$
$I_n$ is given by \emph{Gauss's formula}~\cite{gauss-irr}:
\begin{equation}
\label{eq:gauss}
I_{n} = |\mathcal{I}_{n}| = \frac{1}{n} \sum_{d|n} \mu(d)\cdot
q^{\frac{n}{d}} \enspace ,
\end{equation}
where $d$ ranges over all positive divisors of $n$ (including $1$ and $n$),
while $\mu$ denotes the \emph{M\"{o}bius function}. Let
$d = \varrho_1^{\alpha_1}\varrho_2^{\alpha_2}\cdots\varrho_k^{\alpha_k}$ be the
prime factorization of $d \in \N$. Then, $d$ is called \emph{square-free} (s.f.)
if $\alpha_i = 1$ for all $i \in \{1,\cdots, k\}$, i.e. if $d$ is not divisible
by any prime power with exponent higher than $1$. The M\"{o}bius function of $d$
is defined as:
\begin{equation}
  \label{eq:mobius}
  \mu(d) =
  \begin{cases}
    1  & \mbox{ , if $d=1$ or $d$ is s.f. and has an even number of prime
      factors} \\
    -1 & \mbox{ , if $d$ is s.f. and has an odd number of prime
      factors} \\
    0  & \mbox{ , if $d$ is not s.f.}
  \end{cases}
\end{equation}
We thus have that 
\begin{equation}
\label{eq:first-low-bd}
N_{n} \ge I_{n} \enspace .
\end{equation}
In order to refine this lower bound, we have to determine how many other
(reducible) polynomials of degree $n$ one can add to $\mathcal{I}_{n}$ so that
the resulting set only includes pairwise coprime polynomials. To this end, we
first need a side result which shows that the sequence of the numbers of monic
irreducible polynomials is non-decreasing in the degree $n$. As a preliminary
remark, observe that any polynomial $f$ in $S_n$ is either irreducible and hence
belongs to $\mathcal{I}_n$, or all its irreducible factors belong to
$\mathcal{J}_n = \bigcup_{k =1}^{\lfloor \frac{n}{2} \rfloor} \mathcal{I}_k$.
\begin{lemma}
\label{lm:irr-numb}
For all $q\ge 2$ powers of a prime number and $n\ge 1$, $I_n$ is a
non-decreasing function of $n$.
\end{lemma}
\begin{proof}
  We want to show that $I_n \ge I_{n-1}$ for all $n \ge 2$. Note that
  $I_1 = q-1$ since we do not consider the polynomial $X$ (its constant term
  being null), while $I_n$ is given by Gauss's formula for $n \ge 2$.

  The claim is easily proved for $n \le 4$, since
  \begin{align*}
	I_1 &= q - 1 \enspace , \\
	I_2 &= \frac{1}{2}(q^2-q) = \frac{q}{2}I_1 \ge I_1 \enspace ,\\
	I_3 &= \frac{1}{3}(q^3 - q) = \frac{2(q+1)}{3}I_2 \ge I_2 \enspace ,\\
	I_4 &= \frac{1}{4}(q^4 - q^2) = \frac{3q}{4} I_3 \ge I_3 \enspace .
  \end{align*}
  We now assume $n \ge 5$. We first prove that
  $I_n \ge \frac{1}{n} (q^n - q^{n-2})$. It is easily checked for $n=5$, since
  $I_5 = \frac{1}{5}(q^5 - q)$; for $n \ge 6$, consider the sum
  \begin{equation}
    \label{eq:sum-1}
    q^{\lfloor \frac{n}{2} \rfloor} + q^{\lfloor \frac{n}{2} \rfloor - 1} +
    \cdots + q + 1 = \sum_{i=0}^{\lfloor \frac{n}{2} \rfloor} q^i = \frac{q^{\lfloor
        \frac{n}{2} \rfloor + 1} -1 }{q - 1} \enspace . 
  \end{equation}
  Remark that, for all $d|n$ with $d\neq 1$, the term $q^{\frac{n}{d}}$ in the
  sum of Gauss's formula occurs in the sum of Equation~\eqref{eq:sum-1},
  i.e. for $i = \frac{n}{d}$. Since in Gauss's formula one always adds or
  subtracts the term $q^{\frac{n}{d}}$ depending on the value of $\mu(d)$, by
  Equations~\eqref{eq:gauss} and~\eqref{eq:sum-1} we have the following
  inequality:
  \begin{equation}
    \label{eq:ineq-1}
    I_n \ge \frac{1}{n} \left( q^n - \sum_{i=0}^{\lfloor \frac{n}{2} \rfloor}
      q^i \right) = \frac{1}{n} \left( q^n - \frac{q^{\lfloor \frac{n}{2}
          \rfloor + 1} -1 }{q - 1} \right) \enspace ,
  \end{equation}
  from which it follows that
  \begin{align}
    \nonumber
    I_n &\ge \frac{1}{n} (q^n - q^{\left\lfloor \frac{n}{2} \right\rfloor + 1})\\
    \label{eq:lower-bound-in}
        & \ge \frac{1}{n} (q^n - q^{n-2}) \enspace .
  \end{align}
  We now prove that $I_{n-1} \le \frac{1}{n-1}q^{n-1}$. Similarly to the
  previous inequality, let us denote by $p$ the smallest prime divisor of
  $n-1$. Then, in the sum of Gauss's formula for $I_{n-1}$ we subtract
  $q^{\frac{n-1}{p}}$, since $\mu(p) = -1$. Consider now the sum
  \begin{equation}
    \label{eq:sum-2}
    q^{\frac{n-1}{p} - 1} + q^{\frac{n-1}{p} - 2} +
    \cdots + q + 1 = \sum_{i=0}^{\frac{n-1}{p} - 1} q^i =
    \frac{q^{\frac{n-1}{p}} - 1}{q - 1} \enspace . 
  \end{equation}
  Again, each term $q^{\frac{n-1}{d}}$ in Gauss's formula for $I_{n-1}$ occurs
  in~\eqref{eq:sum-2} for $i = \frac{n-1}{d}$. Thus, the following inequality
  holds:
  \begin{equation}
    I_{n-1} \le \frac{1}{n-1} \left( q^{n-1} - q^{\frac{n-1}{p}} 
             + \frac{q^{\frac{n-1}{p}} - 1}{q - 1} \right) \enspace .
  \end{equation}
  Therefore,
  \begin{equation}
    \label{eq:upper-bound-in}
        I_{n-1} \le \frac{1}{n-1}q^{n-1} \enspace .
  \end{equation}
  Combining the lower bound in~\eqref{eq:lower-bound-in} and the upper bound
  in~\eqref{eq:upper-bound-in}, we obtain
  \begin{align}
	I_n &\ge I_{n-1} \frac{n-1}{n} \cdot \frac{q^n - q^{n-2}}{q^{n-1}}\\
            &=  I_{n-1} \frac{n-1}{n} \cdot \left(q - \frac{1}{q} \right) \enspace .
  \end{align}
  Thus, since $q\ge 2$ and $n>5$, it follows that
  \begin{equation}
	I_n \ge I_{n-1} \frac{4}{5} \cdot \frac{3}{2} \ge I_{n-1} \enspace .
  \end{equation}
\qed
\end{proof}

Consider now the following construction for a family of pairwise coprime
polynomials parameterized on the degree $n$:

\begin{description}
\item[{\sc Construction-Irreducible}$(n)$]
\item[Initialization:] Initialize set $\mathcal{P}_{n}$ to $\mathcal{I}_{n}$
\item[Loop:] For all $1 \le k \le \left \lfloor \frac{n}{2} \right \rfloor$ do:
  \begin{enumerate}
    \item Build set $\mathcal{P}'_{k}$ by multiplying each polynomial in
      $\mathcal{I}_{k}$ with a distinct polynomial in $\mathcal{I}_{n-k}$
    \item Add set $\mathcal{P}'_{k}$ to $\mathcal{P}_{n}$
  \end{enumerate}
\item[Output:] return $\mathcal{P}_{n}$
\end{description}

Hence, set $\mathcal{P}_n$ is constructed by first adding all irreducible
polynomials of degree $n$, then by adding the set of all irreducible polynomials
of degree $1$ multiplied by as many irreducible polynomials of degree $n-1$, and
so on. In particular, notice that step 1 in the loop of {\sc
  Construction-Irreducible} is possible since by Lemma~\ref{lm:irr-numb} one has
that $I_{n-k} \ge I_k$ for all $k\le\left \lfloor \frac{n}{2} \right
\rfloor$. Further, all polynomials added to $\mathcal{P}_n$ through {\sc
  Construction-Irreducible} are pairwise coprime, since they all have distinct
irreducible factors.

Remark that the procedure {\sc Construction-Irreducible} can be iterated only up
to $k\le\left \lfloor \frac{n}{2} \right \rfloor$, because by symmetry the
irreducible polynomials of degree $n-k$ with $k>\left \lfloor \frac{n}{2} \right
\rfloor$ correspond to those of degree $k\le\left \lfloor \frac{n}{2} \right
\rfloor$. Notice also that, when $n$ is even, the last step of the procedure
consists of squaring all irreducible polynomials of degree $\frac{n}{2}$.

Hence, we have shown that the set $\mathcal{P}_n$ which is generated by
procedure {\sc Construction-Irreducible} is indeed a member of the family
$\mathcal{M}_{n}$. The cardinality of such set is given by
\begin{equation}
\label{eq:card-constr}
C_{n} = |\mathcal{P}_{n}| = I_{n} + \sum_{k=1}^{\left \lfloor \frac{n}{2} \right
  \rfloor} I_{k} \enspace .
\end{equation}
In fact, beside the initial step when one adds all irreducible polynomials of
degree $n$ to $\mathcal{P}_n$, in each iteration $k$ of the loop the number of
polynomials that one can obtain by multiplying two irreducible factors is
bounded by the number of irreducible polynomials of degree $k$, which is
$I_{k}$. We have thus obtained the following result, which gives a more precise
lower bound on $N_{n}$:
\begin{equation}
  \label{eq:lower-bound-2}
  N_{n} \ge C_{n} \enspace .
\end{equation}
A natural question arising from Inequality~\eqref{eq:lower-bound-2} is whether
the above construction is optimal, i.e. if the maximum number of pairwise
coprime polynomials $N_{n}$ is actually equal to $C_{n}$. In the next theorem,
we prove that the MOLS produced by {\sc Construction-Irreducible} are indeed
maximal, and we characterize the families of $\mathcal{T}_{n}$.
\begin{theorem}
  \label{thm:char-max-mols}
  For any $n$ and $q$, the maximum number of MOLS generated by LBCA of diameter
  $d = n+1$, or equivalently the maximum number of pairwise coprime monic
  polynomials of degree $n$ and nonzero constant term is:
  \begin{equation}
    \label{eq:char-max-mols}
    N_n = I_n + \sum_{k=1}^{\lfloor \frac{n}{2} \rfloor} I_k \enspace .
  \end{equation}
  Moreover, let $A \subseteq S_n$. Then, $A \in \mathcal{T}_n$ if and only if
  the following hold:
  \begin{enumerate}
  \item $A$ contains $\mathcal{I}_n$;

  \item if $n$ is even then $A$ contains $\{ g^2 : g \in \mathcal{I}_{n/2} \}$;

  \item for every $g \in \mathcal{I}_k$ with $k < n/2$, there exists a unique
    $f \in A$ such that $g | f$. This $f$ is either of the form $f = g^a$ with
    $a = n / k$, or of the form $f = g^b h$ where $b k < n/2$ and
    $h \in \mathcal{I}_{n - bk}$, in which case $h$ does not divide any other
    $f' \in A$.
  \end{enumerate}
\end{theorem}
\begin{proof}
  We first determine the value of $N_n$. By inequality~\eqref{eq:lower-bound-2}
  we have that $N_n \ge C_n$. Conversely, let $A \in \mathcal{M}_n$ be a maximum
  collection of mutually coprime polynomials in $S_n$, i.e. with cardinality
  $N_n$. Clearly, $A$ must contain all irreducible polynomials $\mathcal{I}_n$,
  so let $B = A \setminus \mathcal{I}_n$ be the set of reducible polynomials in
  $A$. For any $f \in B$, denote the irreducible polynomial of lowest degree in
  the factorization of $f$ as $T_f$ (if there are several, choose the first one
  in lexicographic order). Note that $T_f$ has degree at most $n/2$. Now if
  $f,g \in B$ satisfy $T_f = T_g$, then $f$ and $g$ are not coprime, therefore
  the map $f \mapsto T_f$ is an injection from $B$ to $\mathcal{J}_n$. Thus
  $|A| \le C_n$ and $N_n = C_n$.

  We now characterize the families of cardinality $N_n$. We begin by showing
  that any such family must satisfy the three properties of the
  theorem. Firstly, as seen above, such a family $A$ must contain
  $\mathcal{I}_n$, so let us focus on $B = A \setminus \mathcal{I}_n$. This
  time, the mapping $f \mapsto T_f$ is a bijection from $B$ to $\mathcal{J}_n$,
  hence let $g \mapsto F_g$ be its inverse. Secondly, if
  $g \in \mathcal{I}_{n/2}$, then $F_g = gh$ for some $h \in
  \mathcal{I}_{n/2}$. If $h \ne g$, then $F_h \ne F_g$ but $\gcd(F_g, F_h) = h$,
  which violates coprimality; thus $h = g$ and $F_g = g^2$. Thirdly, if
  $g \in \mathcal{I}_k$, then either $F_g = g^a$ for $a = n/k$ or
  $F_g = g^b H_g$ for $bk < n$ and $\gcd(H_g, g) = 1$. If $H_g$ is reducible,
  then its factor of lowest degree $h \in \mathcal{J}_n$ is a common divisor of
  $F_h$ and $F_g$, which again violates coprimality. Thus
  $H_g \in \mathcal{I}_{n - bk}$ for $bk < n/2$. Finally, if $H_g = H_{g'}$ for
  another $g' \in \mathcal{J}_n$, then again coprimality is
  violated. Conversely, it is easily checked that any family satisfying all
  three properties is a family of $N_n$ coprime polynomials in
  $S_n$. \qed
\end{proof}

We now determine how many maximal sets of pairwise coprime polynomials of degree
$n$ one can obtain through {\sc Construction-Irreducible}, thus providing a
lower bound on $T_n$. In particular, this corresponds to the case of families
$A \in \mathcal{T}_n$ where the polynomial $f \in A$ in the third condition of
Theorem~\ref{thm:char-max-mols} is of the form $f = gh$ (i.e.\ $b=1$) and
$h \in \mathcal{I}_{n - bk}$. Moreover, we show that this lower bound is
asymptotically close to the actual value of $T_n$. Before proving this result,
we first need the following asymptotic estimate of $I_n$:

\begin{lemma}
  \label{lm:gauss-est}
  Let $I_n$ be defined as in Equation~\eqref{eq:gauss}. Then, as $n$ tends to
  infinity,
  \begin{equation}
    \label{eq:est}
    I_n = \frac{1}{n} \left(q^n - \mathcal{O}\left(q^{\frac{n}{2}}\right)\right)
    \enspace .
  \end{equation}
\end{lemma}
\begin{proof}
  Let us rewrite Equation~\eqref{eq:gauss} by extracting the terms $d=1$ and
  $d=p$ from the sum, where $p$ is the smallest prime divisor of $n$. Since
  $\mu(1)=1$ and $\mu(p) = -1$, we have
  \begin{equation}
    \label{eq:ext-sum}
    I_n = \frac{1}{n} \left(q^n - q^{\frac{n}{p}} + \sum_{d|n: d\neq 1,p} \mu(d)
      \cdot q^{\frac{n}{d}}\right) \enspace .
  \end{equation}
  The smallest divisor of $n$ which is strictly greater than $p$ can be at most
  $p+1$. Thus, each term in the sum of Equation~\eqref{eq:ext-sum} is limited in
  absolute value by $q^{\frac{n}{p+1}}$. In particular, as $d$ grows the value
  $q^{\frac{n}{d}}$ decreases, hence we can bound the sum in~\eqref{eq:ext-sum}
  with the geometric series $\sum_{i=0}^\infty q^{\frac{n}{p+1}-i}$:
  \begin{equation}
    \label{eq:bound-geom}
    \sum_{d|n: d\neq 1,p} \mu(d) \cdot q^{\frac{n}{d}} \le \sum_{i=0}^\infty
    q^{\frac{n}{p+1}-i} = q^{\frac{n}{p+1}} \sum_{i=0}^\infty q^{-i} \enspace .
  \end{equation}
  Since $q\ge 2$, we have that $\sum_{i=0}^\infty q^{-i} \le 2$. Thus, we obtain
  \begin{equation}
    \label{eq:bound-geom-2}
    \sum_{d|n: d\neq 1,p} \mu(d) \cdot q^{\frac{n}{d}} \le 2\cdot
    q^{\frac{n}{p+1}} \enspace .
  \end{equation}
  Consider now the difference $q^{\frac{n}{p}} - 2\cdot q^{\frac{n}{p+1}}$:
  \begin{equation}
    \label{eq:diff-bound}
    q^{\frac{n}{p}} - 2\cdot q^{\frac{n}{p+1}} = q^{\frac{n}{p}} \left( 1 -
      2\cdot q^{\frac{n}{p+1}-\frac{n}{p}}\right) = q^{\frac{n}{p}} \left( 1 -
      2\cdot q^{-\frac{n}{p(p+1)}}\right) \enspace .
  \end{equation}
  Clearly, $q^{-\frac{n}{p(p+1)}} \rightarrow 0$ for $n \rightarrow
  \infty$. Hence, we have that
  $q^{\frac{n}{p}} - 2\cdot q^{\frac{n}{p+1}} =
  \mathcal{O}\left(q^{\frac{n}{p}}\right)$, and by
  Inequality~\eqref{eq:bound-geom-2} it follows that
  \begin{equation}
    \label{eq:diff-bound-2}
    q^{\frac{n}{p}} - \sum_{d|n: d\neq 1,p} \mu(d) \cdot q^{\frac{n}{d}} =
    \mathcal{O}\left(q^{\frac{n}{p}}\right) \enspace .
  \end{equation}
  Therefore, Equation~\eqref{eq:gauss} can be rewritten as
  \begin{equation}
    I_n = \left(q^n - \mathcal{O}\left(q^{\frac{n}{p}}\right)\right) = \left(q^n
      - \mathcal{O}\left(q^{\frac{n}{2}}\right)\right) \enspace ,
  \end{equation}
  where the rightmost equality follows from the fact that $p\ge 2$ for all
  $n \in \N$. \qed
\end{proof}

We can now prove our lower bound on $T_n$. In what follows, we denote by $D_n$
the number of maximal sets produced by {\sc Construction-Irreducible}.
\begin{theorem}
  \label{thm:lower-bound}
  For all $n$, it holds that
  \begin{displaymath}
    D_n = \prod_{k=1}^{ \lceil \frac{n}{2} - 1 \rceil} \frac{ I_{n-k} ! }{
      (I_{n-k} - I_k)! } \enspace .
  \end{displaymath}
  Moreover, as $n$ tends to infinity, we have
  \begin{align*}
    \log_q D_n &= \Theta(q^{\frac{n}{2}}) \enspace .\\
    \log_q T_n &= \log_q D_n + \mathcal{O}(q^{\frac{n}{3}}) = \Theta( q^{\frac{n}{2}} ) \enspace .
  \end{align*}
\end{theorem}
\begin{proof}
  Let us first prove the formula for $D_n$. For all
  $1 \le k \le \lfloor \frac{n}{2} \rfloor$, the set $\mathcal{P}'_k$ in step 1
  of the loop of {\sc Construction-Irreducible} is obtained by first taking an
  irreducible polynomial $f_1 \in \mathcal{I}_k$ and multiplying it by an
  irreducible polynomial $g_1 \in \mathcal{I}_{n-k}$. Hence, the choices for
  $g_1$ are $I_{n-k}$. Then, one takes another irreducible polynomial
  $f_2 \in \mathcal{I}_k$ and multiplies it by an irreducible polynomial
  $g_2 \in \mathcal{I}_{n-k}$, with $g_2 \neq g_1$. Thus, the possible choices
  for $g_2$ are $I_{n-k}-1$. Since the choices for the polynomials in
  $\mathcal{I}_{n-k}$ are independent, and since we have to select $I_k$ of
  them, we have that the number of choices for constructing $\mathcal{P}_k'$ is
  \begin{equation}
    \label{eq:ek}
    E_k = I_{n-k} (I_{n-k} - 1) \cdots (I_{n-k} - I_k + 1) = \frac{ I_{n-k} ! }{
      (I_{n-k} - I_k)! } \enspace .
  \end{equation}
  Further, since for $1 \le k_1,k_2 \le \lfloor \frac{n}{2} \rfloor$ with
  $k_1\neq k_2$ the choices for constructing $\mathcal{P}_{k_1}'$ and
  $\mathcal{P}_{k_2}'$ are independent, we obtain that
  \begin{equation}
    \label{eq:dn}
    D_n = \prod_{k=1}^{\lceil \frac{n}{2} - 1 \rceil} E_k = \prod_{k=1}^{ \lceil \frac{n}{2} - 1 \rceil} \frac{ I_{n-k} ! }{
      (I_{n-k} - I_k)! } \enspace .
  \end{equation}

  We now prove that $\log_q D_n = \Theta( q^{\frac{n}{2}} )$, starting with some
  estimates for $\log_q I_{n-k}$ and $\log_q(I_{n-k} - I_k)$. As a first remark,
  observe that Equation~\eqref{eq:est} in Lemma~\ref{lm:gauss-est} shows that
  \begin{equation}
    \label{eq:ineq-logink}
    \log_q I_{n-k} \le n-k \enspace .
  \end{equation}
  It is then easy to prove that for $n$ large enough
  and $k < n/2$, one has
  \begin{equation}
    \label{eq:ineq-delta}
    I_k \le \frac{1}{k} q^k < \delta \frac{1}{n-k} q^{n-k} \enspace ,
  \end{equation}
  for $\delta < 1$, e.g.\ $\delta = \frac{1}{q - 1/2}$. Combining
  Inequality~\eqref{eq:ineq-delta} with Equation~\eqref{eq:est} we obtain
  \begin{equation}
    \label{eq:comb-delta-est}
	I_{n-k} - I_k \ge \frac{1}{n-k} \left\{ (1 - \delta) q^{n-k} -
          \mathcal{O}(q^{\frac{n-k}{2}}) \right\} \enspace ,
  \end{equation}
  and hence
  \begin{equation}
    \label{eq:est-log-ink-o1}
    I_{n-k} - I_k = \frac{(1 - \delta - o(1)) q^{n-k}}{n-k} \enspace .
  \end{equation}
  Equation~\eqref{eq:est-log-ink-o1} thus yields the following estimate for
  $\log_q(I_{n-k} - I_k)$:
  \begin{align}
    \nonumber \log_q(I_{n-k} - I_k) &= n - k - \log_q(n-k) + \log_q( 1 - \delta - o(1)) \\
    \label{eq:est-log-i-nk}
                                    &= n - k - \mathcal{O}(\log(n-k)) \enspace .
  \end{align}
  Consider now $\log_qE_k$. By Equation~\eqref{eq:ek}, it is easy to see that
  \begin{equation}
    \label{eq:ineq-ek}
    I_k \log_q (I_{n-k} - I_k) \le \log_q E_k \le I_k \log_q I_{n-k}
    \enspace .
  \end{equation}
  Since by~\eqref{eq:ineq-logink} we have that $\log_q I_{n-k} \le n-k$, while
  by~\eqref{eq:est-log-i-nk} it holds that
  $\log_q E_k \ge I_k ( n - k - \mathcal{O}(\log(n-k)))$, the inequalities
  in~\eqref{eq:ineq-ek} can be rewritten as follows:
  \begin{equation}
    \label{eq:ineq-ek-1}
    I_k ( n - k - \mathcal{O}(\log(n-k))) \le \log_q E_k \le I_k \log_q I_{n-k}
    \enspace .
  \end{equation}
  Consequently, we obtain the following estimate for $\log_q E_k$:
  \begin{align}
    \nonumber \log_q E_k &= I_k \left( n - k - \mathcal{O}(\log(n-k)) \right)\\
    \nonumber           &= \frac{1}{k} (q^k - \mathcal{O}(q^{\frac{k}{2}}) ) \left( n -
                          k - \mathcal{O}(\log
                          (n-k)) \right)\\
    \nonumber           &= \frac{1}{k} ( (1 - o(1)) q^k ) ( ( 1 - o(1)) (n - k )
                          ) \\
    \nonumber           &= \frac{1}{k} (1 - o(1)) q^k (n-k) \\
    \label{eq:est-logek}
                         &= \frac{n-k}{k} q^k - o\left(\frac{n}{k} q^k\right)
                           \enspace .
  \end{align}
  Denoting $\sigma = \sum_{i=0}^\infty q^{-i}$, we have
  \begin{equation}
    \label{eq:ineq-geom}
    q^{\lceil \frac{n}{2} - 1 \rceil} \le \sum_{k= \left\lfloor \frac{n}{3} \right\rfloor + 1}^{
      \lceil \frac{n}{2} - 1 \rceil} \frac{n-k}{k} q^k \le 2 \sigma q^{\lceil
      \frac{n}{2} - 1 \rceil} \enspace .
  \end{equation}
  Therefore,
  \begin{align}
    \nonumber \log_q D_n &= \sum_{k=1}^{ \lceil \frac{n}{2} - 1 \rceil} \log_q E_k\\
    \nonumber            &= \sum_{k= \left\lfloor \frac{n}{3} \right\rfloor + 1}^{ \lceil
                           \frac{n}{2} - 1 \rceil} \log_q E_k +
                           \sum_{k=1}^{\left\lfloor \frac{n}{3} \right\rfloor} \log_q E_k\\
    \label{eq:est-logdn}
                         &= \Theta(q^{\frac{n}{2}}) + \mathcal{O}( n q^{\frac{n}{3}} ) \enspace .
  \end{align}

  We finally prove that $\log_q T_n = \log_q D_n + \mathcal{O}(q^{\frac{n}{3}})$. By
  Theorem~\ref{thm:char-max-mols}, in any maximal family of polynomials in
  $\mathcal{T}_n$, and any irreducible $g \in \mathcal{I}_k$ for
  $n/3 < k < n/2$, the corresponding $f$ must be $f = g h$ for some
  $h \in \mathcal{I}_{n-k}$, thus there are $E_k$ choices for the polynomials in
  the family that have an irreducible factor of degree $k$. If $k \le n/3$, then
  for any $g \in \mathcal{I}_k$ there are at most
  $1 + \sum_{d=\left\lfloor \frac{n}{2} \right\rfloor + 1}^{n-k} I_d$ choices for the corresponding
  polynomial $f$ in the family. Altogether, we obtain
  \begin{equation}
    \label{eq:ineq-tn}
    T_n \le \prod_{k= \left\lfloor \frac{n}{3} \right\rfloor + 1}^{ \lceil \frac{n}{2} - 1 \rceil}
    E_k \cdot \prod_{k= 1}^{\left\lfloor \frac{n}{3} \right\rfloor} \left\{ 1 + \sum_{d=\left\lfloor \frac{n}{2}
        \right\rfloor + 1}^{n-k} I_d \right\}^{I_k} \enspace .
  \end{equation}
  Define now $B_k$ as
  \begin{equation}
    \label{eq:bk}
	B_k = \left\{ 1 + \sum_{d=\left\lfloor \frac{n}{2} \right\rfloor + 1}^{n-k} I_d
        \right\}^{I_k} \le \left\{ q^{n-k} \right\}^{ \frac{1}{k} q^k } \le
        q^{\frac{n-k}{k} q^k} \enspace .
  \end{equation}
  Again, this yields
  $\sum_{k=1}^{\left\lfloor \frac{n}{3} \right\rfloor} \log_q B_k = \mathcal{O}(q^{\frac{n}{3}})$, which in
  turn gives us
  \begin{equation}
    \log_q T_n \le \log_q D_n + \sum_{k=1}^{\left\lfloor \frac{n}{2} \right\rfloor} \log_q B_k =
    \log_q D_n + \mathcal{O}( q^{\frac{n}{3}} ) \enspace .
  \end{equation}
  \qed
\end{proof}

\section{Conclusions and Perspectives}
\label{sec:conc}
In this paper we undertook an investigation of mutually orthogonal Latin squares
generated through linear bipermutive CA. First, we proved that any bipermutive
CA of diameter $d$ and length $2(d-1)$ can be used to generate a Latin square of
order $N = q^{d-1}$, with $q$ being the size of the CA alphabet. We then focused
on orthogonal Latin squares generated by LBCA, showing a characterization result
based on the Sylvester matrix induced by two linear local rules. In particular,
we proved that two LBCA generate orthogonal Latin squares if and only if the
polynomials associated to their local rules are relatively prime. In the second
part of the paper, we determined the number of LBCA pairs over $\F_q$ generating
orthogonal Latin squares, i.e. the number of coprime polynomial pairs $(f,g)$ of
degree $n$ over $\F_q$ where both $f$ and $g$ have nonzero constant term. In
particular, we remarked that the integer sequence generated by the closed-form
formula of the recurrence equation for $q=2$ corresponds to {\sc A002450}, a
sequence which is already known in the OEIS for several other facts not related
to polynomials or orthogonal Latin squares. In the last part of the paper, we
presented a construction of MOLS generated by LBCA based on irreducible
polynomials, and we proved that the resulting families are maximal. Finally, we
also derived a lower bound on the number of maximal MOLS induced by the proposed
construction, and showed that this bound is asymptotically close to the actual
number of maximal MOLS produced by LBCA.

There are several opportunities for further improvements on the results
presented in this paper.

A first direction for future research is to generalize the study to MOLS
generated by \emph{nonlinear} bipermutive CA. In this case, one obviously cannot
rely on the characterization result of Theorem~\ref{thm:ols-lbca}, since this
crucially depends on the use of the Sylvester matrix defined from the transition
matrices of linear CA. Preliminary work led by some of the authors of the
present paper showed that a necessary condition for a pair of BCA (either linear
or nonlinear) to generate orthogonal Latin squares is that their local rules
must be \emph{pairwise balanced}, meaning that each of the four pair of bits
must occur equally often in the juxtaposition of their truth
tables~\cite{mariot17-automata}. We believe it is still possible to use the
theory of resultants to characterize orthogonal Latin squares generated by
nonlinear BCA. As a matter of fact, the main difference between linear and
nonlinear pairs is that in the former case the system of
equations~\eqref{eq:system-1} concerns \emph{univariate} polynomials. On the
other hand, in the nonlinear case one can associate \emph{multivariate}
polynomials to the local rules, and then use the tools of elimination theory (to
which the concept of resultant belongs) to study the invertibility of the
resulting systems.

A second extension worth exploring, especially concerning the possible
applications related to secret sharing, is to investigate the structure of the
inverse of a Sylvester matrix. As described in~\cite{mariot16}, a family of $k$
MOLS generated by LBCA can be used to design a $(2,k)$-threshold secret where
the dealing phase corresponds to evaluating the global rules of the $k$ LBCA to
an initial configuration whose left half is the secret, while the right half is
randomly chosen. The outputs of the LBCA will be the shares distributed to the
$k$ players. In order to reconstruct the secret, any two out of $k$ players must
invert the Sylvester matrix associated to their CA (which are assumed to be
public) and then multiply it by the vector obtained by concatenating their
shares. Hence, an interesting question is whether the reconstruction phase can
be carried out again by CA computation, which means that the inverse of the
Sylvester matrix related to two LBCA must be of Sylvester type as well. This
question has been answered in negative during the Fifth International Students'
Olympiad in Cryptography -- NSUCRYPTO~\cite{olympiad-2018} for LBCA over the
finite field $\F_2$. In particular, it has been proved that the only Sylvester
matrix over $\F_2$ satisfying this condition is the one defined by the
polynomials $X^n$ and $1+X^n$, which does not correspond to a pair of LBCA since
$X^n$ has null constant term. However, the existence question for Sylvester
matrices whose inverses are of Sylvester type remains open for larger finite
fields.

Finally, another interesting idea would be to extend our investigation to
\emph{Mutually Orthogonal Latin Hypercubes} generated by CA, i.e. the
generalization of MOLS to higher dimensions. This would be equivalent to study
the conditions under which CA can be used to construct orthogonal arrays with
strength higher than $2$. A characterization result for such kind of orthogonal
arrays would allow one to design a general $(t,n)$-threshold secret sharing
scheme based on CA, or equivalently to design linear MDS codes through CA. A
possible idea to achieve this result would be to first characterize which
subclasses of bipermutive CA generate \emph{Latin hypercubes}. From there, the
next step would be to characterize sets of linear CA inducing Orthogonal Latin
Hypercubes, which are equivalent to orthogonal
arrays~\cite{keedwell15}. However, we note that there are no straightforward
ways to generalize the concept of resultant to more than two
polynomials~\cite{gelfand}. As a matter of fact, some of the existing
generalizations involve matrices which do not correspond to those related to
hypercubes generated by CA. To the best of our knowledge, the only resultant
matrix for several polynomials that most resemble the CA hypercube case has been
defined in~\cite{deissler}, which could thus represent a starting point for
future work on the subject.

\subsection*{Acknowledgements}
The authors wish to thank Arthur Benjamin, Curtis Bennett and Igor Shparlinski
for insightful comments on how to count the number of pairs of coprime
polynomials with nonzero constant term.

\bibliographystyle{spmpsci}
\bibliography{bibliography}

\end{document}